%% file: SysDO2024_adaptive_MPC.tex
\documentclass{svproc}
\usepackage{amsmath,amssymb}    % math packages, e.g. for align and pmatrix environments
\usepackage{optidef} % for minimization problems 
\usepackage{algorithm}
\usepackage{bbold} % for \mathbb{1}
\usepackage{constants} % to automatically count the constants; 

\usepackage{float}
\usepackage{enumerate}     
\usepackage{subcaption}
\usepackage{color}
\usepackage{url}

\usepackage{pgfplots}
\pgfplotsset{compat=newest}
\usetikzlibrary{plotmarks}
\usetikzlibrary{arrows.meta}
\usepgfplotslibrary{patchplots}
\usepackage{grffile}

\usepackage{comment}
\newtheorem{assumption}{Assumption}
\newtheorem{prop}{Proposition}

\newtheorem{cor}{Corollary}

\newcommand{\R}{\mathbb{R}}
\newcommand{\K}{\mathcal{K}}

\newcommand{\U}{\mathbb{U}}
\newcommand{\X}{\mathbb{X}}

\newcommand{\Zs}{\mathcal{Z}^{\mathrm{s}}}
\newcommand{\ZsAPPR}[1]{\hat \Zs\lr{\D_{#1}}}
\newcommand{\ZsAPPRy}[1]{\hat{\mathcal{Z}}_{\mathrm{y}}^\mathrm{s}\lr{\D_{#1}}}

\newcommand{\Jmpc}{J_{\mathrm{MPC}}}

\newcommand{\costFuncApprOPT}[1]{\hat J^*\lr{\D_{#1}}}

\newcommand{\costFuncMpcOPT}[1]{\Jmpc^*\lr{x_{#1},\; \D_{#1}}}
\newcommand{\LyapFunc}[1]{V\lr{x_{#1},\D_{#1}}}

\newcommand{\costFuncMAX}{\hat J^{\max}}
\newcommand{\costFuncMAXmpc}{\hat J^{\max}_{\mathrm{MPC}}}

\newcommand{\at}[2]{\left.#1\right|_{#2}}
\newcommand{\norm}[1]{\left |\left | #1 \right |\right |}

\newcommand{\set}[2]{\left\{ #1 \middle|\, #2\right\}}
\newcommand{\lr}[1]{\left( #1 \right)}

\newcommand{\mr}[1]{\mathrm{#1}}

\newcommand{\D}{\mathcal{D}}
\newcommand{\bone}{\mathbb{1}}
\newcommand{\dtilde}[1]{\tilde{\raisebox{0pt}[.95\height]{$\tilde{#1}$}}}

\newcommand{\Aidx}[1]{\D_{#1}} % 'couse it changed so often...

\newconstantfamily{ceps}{symbol=c}
\newconstantfamily{cV}{symbol=\tilde c}
\newconstantfamily{oC}{symbol=\bar c}
\usepackage{url}

\begin{document}
\mainmatter              % start of a contribution
\title{Adaptive tracking MPC for nonlinear systems via online linear system identification}
\titlerunning{Adaptive tracking MPC for nonlinear systems}  % abbreviated title (for running head)
\author{Tatiana Strelnikova\inst{1} \and Johannes K\"ohler\inst{2}
\and Julian Berberich\inst{3}}
\authorrunning{T. Strelnikova et al.} % abbreviated author list (for running head)
\tocauthor{Tatiana Strelnikova, Johannes K\"ohler, and Julian Berberich}
\institute{M.Sc.\ student, University of Stuttgart, 70569 Stuttgart, Germany,\\
\email{strelnikovatatyana@gmail.com}
\and
Institute for Dynamic Systems and Control, ETH Z\"urich, 8092 Z\"urich, Switzerland,\\
\email{jkoehle@ethz.ch}
\and
University of Stuttgart, Institute for Systems Theory and Automatic Control, 70569 Stuttgart, Germany,
\email{julian.berberich@ist.uni-stuttgart.de}
}

\maketitle              % typeset the title of the contribution

\begin{abstract}
This paper presents an adaptive tracking model predictive control (MPC) scheme to control unknown nonlinear systems based on an adaptively estimated linear model.
The model is determined based on linear system identification using a moving window of past measurements, and it serves as a local approximation of the underlying nonlinear dynamics.
We prove that the presented scheme ensures practical exponential stability of the (unknown) optimal reachable equilibrium for a given output setpoint.
Finally, we apply the proposed scheme in simulation and compare it to an alternative direct data-driven MPC scheme based on the Fundamental Lemma.
\end{abstract}
%

% %%%%%%%%%%%%%%%%%%%%%%%%%%%%%%%%%%%%%%%%%%%%%%%%%%%%%%% INCLUDE %%%%%%%%%%%%%%%%%%%%%%%%%%%%%%%%%%%%%%%%%%%%%%%%%%%%%%%

\section{Introduction}
Designing controllers for unknown systems based on data is a subject of increasing interest in the recent literature.
One popular approach is to employ the Fundamental Lemma~\cite{willems2005note} for parametrizing trajectories using directly measured data.
Model predictive control (MPC) is a modern control technique which can handle nonlinear systems, performance criteria, and constraints on system variables~\cite{rawlings2017model}.
One of the main applications of the Fundamental Lemma is the design of direct data-driven MPC schemes for linear systems~\cite{berberich2021guarantees,coulson2019deepc}, see~\cite{faulwasser2023behavioral,markovsky2021behavioral,markovsky2023data,verheijen2023handbook} for recent survey papers.
These MPC schemes even admit desirable closed-loop guarantees in the presence of noise~\cite{berberich2021guarantees}.

Alternatively, data can be used in an indirect data-driven approach, determining a model via system identification~\cite{ljung1987system} which can then be employed for model-based MPC.
For example, in the field of adaptive MPC, set membership techniques have been employed to robustly control uncertain systems while additionally reducing uncertainty based on online measurements~\cite{adetola2011robust,guay2015robust,lu2021robust,tanaskovic2014adaptive}.
Despite the surge of research on direct data-driven MPC, the connection between direct and indirect data-driven MPC approaches is still subject to ongoing research~\cite{breschi2023data,doerfler2023bridging,doerfler2023data,krishnan2021on}.

In this paper, we present an indirect data-driven MPC scheme for controlling unknown nonlinear systems based on measured data.
The proposed approach uses a moving window of online measurements in order to estimate an approximate local model of the nonlinear system via linear system identification.
We prove that the resulting indirect data-driven MPC approach practically exponentially stabilizes the closed-loop system.
The presented approach is inspired by and closely follows the recent works~\cite{berberich2022linear1,berberich2022linear2}.
The main difference is that~\cite{berberich2022linear1} uses the true linearized dynamics for prediction (i.e., it is a model-based MPC scheme) and~\cite{berberich2022linear2} uses only input-output data via the Fundamental Lemma~\cite{willems2005note} for prediction (i.e., it is a direct data-driven MPC scheme).
We show that using an indirect data-driven control approach leads to comparable theoretical guarantees as in~\cite{berberich2022linear2}.
Further, we compare the performance of both approaches in a numerical case study, highlighting individual advantages and drawbacks.

The remainder of the paper is structured as follows.
In Section~\ref{sec:preliminaries}, we present the problem setup as well as the considered system identification approach.
Section~\ref{sec:scheme} introduces the proposed indirect data-driven MPC scheme, whose theoretical guarantees are proven in Section~\ref{sec:analysis}.
Further, in Section~\ref{sec:example}, we apply the proposed approach to an example system from the nonlinear MPC literature and compare the outcome to the approaches from~\cite{berberich2022linear1,berberich2022linear2}.
Finally, Section~\ref{sec:conclusion} concludes the paper with a discussion of advantages and drawbacks of the proposed indirect data-driven MPC scheme in comparison to the direct data-driven approach from~\cite{berberich2022linear2}.

\textit{Notation:}
For a vector $x$ and a matrix $P=P^\top$, we write $\| x\rVert_2$ for the Euclidean norm of $x$ and $\| x\rVert_P:=\sqrt{x^\top Px}$ for the weighted norm of $x$.
Moreover, we denote the minimum eigenvalue and singular value of $P$ by $\lambda_{\min}(P)$ and $\sigma_{\min}(P)$, respectively.
We write $\mathcal{K}_{\infty}$ for the set of all continuous functions $\alpha:\mathbb{R}_{\geq0}\to\mathbb{R}_{\geq0}$ which are strictly increasing, unbounded, and satisfy $\alpha(0)=0$.

\section{Preliminaries}\label{sec:preliminaries}
In this section, we introduce the problem setup (Section~\ref{subsec:problem_setup}) as well as the employed system identification approach (Section~\ref{subsec:system_identification}).

\subsection{Problem setup}\label{subsec:problem_setup}

We consider a discrete-time, control-affine nonlinear system of the form
\begin{align}
	\label{eq:nonlinSys}
	\begin{split}
		x_{t+1} &= f(x_t, u_t) = f_0(x_t) + Bu_t, \\
		y_t     &= h(x_t, y_t) = h_0(x_t) + D u_t,
	\end{split}
\end{align}
with state $x_t \in \R^n$, control input $u_t \in \R^m$, output $y_t \in \R^p$, and time $t \in \mathbb{N}$. 
The vector fields $f_0: \R^n \to \R^n$, $h_0: \R^n \to \R^p$ are twice continuously differentiable and $B \in \R^{m \times n}$ and $D \in \R^{p \times m}$ are real matrices.  The system is subject to input constraints $u_t\in\mathbb{U}$, $t\in\mathbb{N}$, with a compact set $\U\subseteq\mathbb{R}^m$.
The  goal is to steer the nonlinear system \eqref{eq:nonlinSys} towards a steady-state that tracks an output reference $y^{\mathrm{r}} \in \R^p$, i.e., $\lim_{t\rightarrow\infty}\|y_t-y^{\mathrm{r}}\|_2=0$.
The key challenge is that the system dynamics~\eqref{eq:nonlinSys} are unknown, i.e., $f_0$, $h_0$, $B$, and $D$ are not available, and instead we only have access to state and output measurements of the system~\eqref{eq:nonlinSys}.

\subsection{System identification}\label{subsec:system_identification}
The proposed approach relies on the online identification of an affine model that locally approximates the nonlinear dynamics~\eqref{eq:nonlinSys}.
In particular, for a state $\tilde{x}\in\mathbb{R}^n$ and time $t\in\mathbb{N}$, we  define the affine dynamics 
\begin{align} \label{eq:linSys}
	\begin{split}
	x_{t+1} &= f_{\tilde{x}} (x_t, u_t) := A_{\tilde{x}} x_t + B u_t + e_{\tilde{x}}\\
	y_t     &= h_{\tilde{x}} (x_t, u_t) := C_{\tilde{x}} x_t + D u_t + r_{\tilde{x}}
	\end{split}
\end{align}
where 
\begin{align*}
	A_{\tilde{x}} = \at{\frac{\partial f_0}{\partial x} }{\tilde{x}}, \;\ e_{\tilde{x}} := f_0(\tilde{x}) - A_{\tilde{x}} x_t, 
	\\
	C_{\tilde{x}} = \at{\frac{\partial h_0}{\partial x} }{\tilde{x}}, \;\ r_{\tilde{x}} := h_0(\tilde{x}) - C_{\tilde{x}} x_t,
\end{align*}
%
%System~\eqref{eq:linSys} 
which corresponds to linearizing the nonlinear system \eqref{eq:nonlinSys} at $(\tilde{x}, 0)$.
We use input-state-output measurements to identify these affine dynamics~\eqref{eq:linSys} online and, thereby, provide a local approximation of the nonlinear dynamics~\eqref{eq:nonlinSys}.
More precisely, at time $t$, we estimate the dynamics linearized at $x_t$ using the most recent $N$ measurements of the nonlinear system~\eqref{eq:nonlinSys}, i.e., using the data set
$\D_t:=\big\{\{x_k\}_{k=t-N}^t, \{u_k, y_k\}_{k=t-N}^{t-1}\big\}$.

The main rationale behind this approach is that, if the closed-loop trajectory does not change too rapidly, then the data $\D_t$ can be approximately explained from the linearized dynamics~\eqref{eq:linSys}.
Thus, they allow to identify the linearized dynamics via least-squares estimation:
Given the data set $\D_t$ at time $t$, we determine an estimate $(\hat A_{t}, \hat B_t, \hat e_{t})$ of the linearized dynamics $(A_{x_t},B,e_{x_t})$ as the minimizer of 
	\begin{align}
		\min_{\tilde A, \tilde B, \tilde e}
		\sum_{k=t-N}^{t-1} \norm{x_{k+1} - \tilde A x_k - \tilde B u_k - \tilde e}^2_2.
		\label{eq:LSPbm}
	\end{align}
Likewise, an estimate $(\hat C_{t}, \hat D_t, \hat r_{t})$ of the linearized output parameters $(C_{x_t},D,r_{x_t})$ is determined by solving
	\begin{align}
		\min_{\tilde C, \tilde D, \tilde r}
		\sum_{k=t-N}^{t-1} \norm{y_{k} - \tilde C x_k - \tilde D u_k - \tilde r}^2_2.
		\label{eq:LSPbm_output}
	\end{align}
Let us denote 
\begin{align*}
	\begin{aligned}
		X_t &:= [x_{t-N}, x_{t-N+1}, \dots, x_{t-1}], 
		&U_t:=& [u_{t-N}, u_{t-N+1}, \dots, u_{t-1}], \\
		X_t^+ &:= [x_{t-N+1}, x_{t-N+2}, \dots, x_{t}], 
		&Y_t:=& [y_{t-N}, y_{t-N+1}, \dots, y_{t}],
	\end{aligned}
\end{align*} 
and $Z_t:=\begin{bmatrix}
			X_t^{\top} & U_t^{\top} & \bone \end{bmatrix}^{\top}$ with $\bone = \begin{bmatrix}
	1 & 1 & \dots & 1
\end{bmatrix}^{\top}$. 
If $Z_t$ has full row rank (compare Assumption~\ref{ass:PE} below), then the minimizers of~\eqref{eq:LSPbm} and~\eqref{eq:LSPbm_output} are given by
\begin{align}\label{eq:lsq_explicit_solution}
	\begin{bmatrix}
		\hat{A}_t & \hat{B}_t & \hat{e}_t
	\end{bmatrix} &= X_t^+ Z_t^{\top} (Z_t Z_t^{\top})^{-1},\>\>
	\begin{bmatrix}
		\hat{C}_t & \hat{D}_t & \hat{r}_t
	\end{bmatrix} = Y_t Z_t^{\top} (Z_t Z_t^{\top})^{-1}.
\end{align}

\begin{comment}
If $Z_t:=\begin{bmatrix}
			X_t^{\top} & U_t^{\top} & \bone \end{bmatrix}^{\top}$ with $\bone = \begin{bmatrix}
	1 & 1 & \dots & 1
\end{bmatrix}^{\top}$ has full row rank (compare Assumption~\ref{ass:PE} below), then the minimizers of~\eqref{eq:LSPbm} and~\eqref{eq:LSPbm_output} are given by
\begin{align*}
	\begin{bmatrix}
		\hat{A}_t & \hat{B}_t & \hat{e}_t
	\end{bmatrix} &= X_t^+ Z_t^{\top} (Z_t Z_t^{\top})^{-1},\>\>
	\begin{bmatrix}
		\hat{C}_t & \hat{D}_t & \hat{r}_t
	\end{bmatrix} = Y_t Z_t^{\top} (Z_t Z_t^{\top})^{-1},
\end{align*}
%
where
%
\begin{align*}
	\begin{aligned}
		X_t &:= [x_t, x_{t+1}, \dots, x_{t+N-1}], 
		&U_t:=& [u_t, u_{t+1}, \dots, u_{t+N-1}], \\
		X_t^+ &:= [x_{t+1}, x_{t+2}, \dots, x_{t+N}], 
		&Y_t:=& [y_t, y_{t+1}, \dots, y_{t+N-1}],
	\end{aligned}
\end{align*}
\end{comment}
%
Let us now define the identified dynamics for the given data $\D_t$ as
\begin{align}\label{eq:identSys}
	\begin{split}
		x_{k+1} &= \hat f_{t}(x_k,u_k) :=\hat{A}_tx_k  +\hat{B}_tu_k +\hat{e}_t, \\
		y_k     &=\hat  h_{t}(x_k,u_k) :=\hat{C}_tx_k +\hat{D}_t u_k +\hat{r}_t,\;\; k \in\mathbb{N}.	\end{split}
\end{align}

\section{Adaptive tracking MPC for nonlinear systems via linear system identification}\label{sec:scheme}

In this section, we present the proposed linear tracking MPC scheme that steers the nonlinear system \eqref{eq:nonlinSys} towards a desired target setpoint $y^{\mathrm{r}}$.
This scheme relies on online system identification using least-squares estimation and, as shown in Section~\ref{sec:analysis}, admits desirable closed-loop guarantees.

At time $t \geq N$, the MPC is based on the following optimization problem
\begin{subequations}
	\label{eq:myMPC}
	\begin{align}
		\underset{{\begin{array}{c}
					\hat x(t), \hat u(t), \\
					\hat x^{\mathrm{s}}(t), \hat u^{\mathrm{s}}(t),
					\hat y^{\mathrm{s}}(t)
		\end{array}}}{\mathrm{minimize}} 
		&\Jmpc\lr{\hat x(t), \hat u(t), \hat x^{\mathrm{s}}(t), \hat u^{\mathrm{s}}(t), \hat y^{\mathrm{s}}(t)} 
		%		\sum_{k = 0}^{L-1} \norm{\hat x_k(t)  - x^{\mathrm{s}}(t)}^2_Q
		%		+
		%		\norm{\hat u_k(t) - u^{\mathrm{s}}(t)}_R^2 + \norm{y^{\mathrm{s}}(t) - y^{\mathrm{r}}}_S^2 
		\\
		\text{ subject to } 
		&\hat x_{k+1}(t) = \hat{A}_t\hat x_k(t) + \hat{B}_t \hat u_k(t) + \hat{e}_t,\label{eq:myMPC_model}
		\\\label{eq:myMPC_init_term}
		& \hat x_0(t) = x_t, \hat x_N(t) = \hat x^{\mathrm{s}}(t),\\\label{eq:myMPC_steady_state}
		&
  \hat{x}^{\mathrm{s}}(t)
  =\hat{A}_t\hat{x}^{\mathrm{s}}(t)
  +\hat{B}_t\hat{u}^{\mathrm{s}}(t) +\hat{e}_t,\\\label{eq:myMPC_steady_state_output}
		&\hat y^{\mathrm{s}}(t) = \hat{C}_t \hat x^{\mathrm{s}}(t) + \hat{D}_t \hat u^{\mathrm{s}}(t) + \hat{r}_t,
  \\\label{eq:myMPC_input_constraints}
		&\hat u_k(t) \in \mathbb{U}, k = 0 , 1, \dots, L-1,\> \hat{u}^{\mathrm{s}}(t)\in\U^{\mathrm{s}},
	\end{align}
\end{subequations}
where the MPC cost function is defined as
\begin{align} \label{eq:costFuncMPC}
	\begin{split}
		&\Jmpc\lr{\hat x(t), \hat u(t), \hat x^{\mathrm{s}}(t), \hat u^{\mathrm{s}}(t), \hat y^{\mathrm{s}}(t)}  \\
		:= &\sum_{k = 0}^{L-1} \norm{\hat x_k(t)  -\hat  x^{\mathrm{s}}(t)}^2_Q
		+
		\norm{\hat u_k(t) - \hat u^{\mathrm{s}}(t)}_R^2 + \norm{\hat y^{\mathrm{s}}(t) - y^{\mathrm{r}}}_S^2
	\end{split}
\end{align}
with user-specified, positive definite weight matrices $Q$, $R$, $S$.
Here, $\hat x(t)$  and $\hat u(t)$ denote the state and input trajectory which are predicted at time $t$ over the horizon $L$.
The prediction model~\eqref{eq:myMPC_model} relies on the identified dynamics~\eqref{eq:identSys} which serve as an approximation of the local linearization~\eqref{eq:linSys} at $x_t$.
Further,~\eqref{eq:myMPC_init_term} contains initial conditions as well as a terminal equality constraint which ensures that, at the end of the prediction horizon, the predicted state is equal to \emph{some} (artificial) steady-state $\hat{x}^{\mathrm{s}}(t)$.
To be precise, $\hat{x}^{\mathrm{s}}(t)$, $\hat{u}^{\mathrm{s}}(t)$, and $\hat{y}^{\mathrm{s}}(t)$ are online optimization variables which, according to~\eqref{eq:myMPC_steady_state} and~\eqref{eq:myMPC_steady_state_output}, are steady-states for the identified dynamics.
The cost~\eqref{eq:costFuncMPC} ensures that the predicted trajectory is kept close to the artificial steady-state (first and second term in~\eqref{eq:costFuncMPC}), whose distance to the setpoint $y^{\mathrm{r}}$ is penalized (third term in~\eqref{eq:costFuncMPC}).
Online optimization over artificial steady-states can substantially increase the region of attraction and allows MPC schemes to cope with online setpoint changes~\cite{limon2008mpc}.
Finally,~\eqref{eq:myMPC_input_constraints} contains the desired input constraints as well as a constraint on the artificial equilibrium input for some convex and compact $\U^{\mathrm{s}}\subset\mathrm{int}(\U)$, which is required for a technical argument in the theoretical analysis below.

If $\U$, $\U^{\mathrm{s}}$ are polytopic, then problem~\eqref{eq:myMPC} is a convex quadratic program which can be solved efficiently.
We denote the minimizer of~\eqref{eq:myMPC} by $\hat x(t)^*$, $\hat u^*(t)$, $\hat x^{\mathrm{s}*}(t)$, $\hat u^{\mathrm{s}*}(t)$, $\hat y^{\mathrm{s}*}(t)$ and the corresponding optimal cost by $\costFuncMpcOPT{t}$.

Algorithm~\ref{alg:MPC} summarizes the proposed MPC scheme.
The algorithm consists of two main steps, indicating its \emph{indirect} data-driven control nature:
At time $t$, the last $N$ input-state-output measurements are used to estimate an affine prediction model~\eqref{eq:linSys} which serves as a local approximation of the nonlinear dynamics~\eqref{eq:nonlinSys}.
Next, the MPC problem~\eqref{eq:myMPC} is solved and the first $n$ steps of the optimal input trajectory are applied to the nonlinear system.
We consider such a multi-step implementation to ensure recursive feasibility despite the model approximation error and the terminal equality constraints, see~\cite{berberich2022inherent}.

\begin{algorithm}\caption{
Adaptive tracking MPC scheme
} \label{alg:MPC}
	\textbf{Offline:} Choose $Q$,$R$,$S$, $\U$, $L$, $y^{\mathrm{r}}$, and
	generate initial data trajectory $\D_N$ of length $N$.
	\begin{enumerate}
		\item At time $t$ and given data $\D_t$, compute $\hat{A}_t$, $\hat{B}_t$, $\hat{e}_t$, $\hat{C}_t$, $\hat{D}_t$, $\hat{r}_t$ as in~\eqref{eq:lsq_explicit_solution}.
		\label{item:step1}
		\item Solve problem \eqref{eq:myMPC} and apply $u_{t+k} = \hat u^*_k(t)$, $k = 0, 1, \dots, n-1$, over the next $n$ time steps.
		\item Set $t=t+n$ and go back to \ref{item:step1}.
	\end{enumerate}
\end{algorithm}

We note that Algorithm~\ref{alg:MPC} is analogous to the MPC approach in~\cite{berberich2022linear1} with the main difference that, instead of the true linearized dynamics, a data-driven estimate thereof is used for prediction.
This makes the proposed approach an indirect data-driven control method.
Alternatively,~\cite{berberich2022linear2} employs an analogous MPC scheme using the Fundamental Lemma~\cite{willems2005note} for prediction, which leads to a direct data-driven control method.
In Section~\ref{sec:conclusion}, we discuss differences of the proposed approach to~\cite{berberich2022linear2}, in particular advantages and drawbacks. 

\section{Theoretical analysis}\label{sec:analysis}
In this section, we show that Algorithm~\ref{alg:MPC} ensures that the closed-loop system practically exponentially stabilizes the (unknown) steady-state that tracks the desired reference $y^{\mathrm{r}}$.
The main idea is as follows:
If the cost matrix $S$ is sufficiently small and the current state $x_t$ is close to the steady-state manifold, then minimization of the cost~\eqref{eq:costFuncMPC} implies that the optimal artificial steady-state $\hat{x}^{s*}(t)$ is close to $x_t$.
In particular, the entire predicted trajectory $\hat{x}(t)$ remains in a small region around $x_t$, where the dynamics linearized at $x_t$~\eqref{eq:linSys} are a good approximation of the nonlinear dynamics~\eqref{eq:nonlinSys}.
Similarly, in closed-loop operation, the past data trajectory $\D_t$ varies within a small region around $x_t$ and, therefore, the least-squares estimate used for prediction is a close approximation of the linearized dynamics.
In summary, a small choice of $S$ ensures that the closed loop moves towards $y^{\mathrm{r}}$ \emph{slowly enough}
such that the identified linear model is a sufficiently close approximation of the nonlinear dynamics. 

In the remainder of this section, we will make this argument more precise by proving practical exponential stability under the proposed MPC scheme.
In Section~\ref{subsec:assumptions}, we first introduce technical assumptions and definitions concerning the system dynamics as well as the available data that are required for the theoretical results.
We note that these assumptions are analogous to the assumptions made on model-based and data-driven linearization-based MPC for nonlinear systems in~\cite{berberich2022linear1,berberich2022linear2}.
Next, in Section~\ref{subsec:stability} we present our main stability result.

\subsection{Technical assumptions and definitions}\label{subsec:assumptions}

The steady-state manifold of~\eqref{eq:nonlinSys} is given by 
	\begin{align}
		\Zs &:= \set{\begin{pmatrix} x^{\mathrm{s}}, u^{\mathrm{s}}\end{pmatrix} \in \R^n \times \U^{\mathrm{s}}}{x^{\mathrm{s}} = f(x^{\mathrm{s}}, u^{\mathrm{s}})}.
	\end{align}
 We define the projection on the output as
	\begin{align}
		\Zs_y &:=  \set{y^{\mathrm{s}} \in \R^p}{\exists (x^{\mathrm{s}}, u^{\mathrm{s}})\in\Zs: \, y^{\mathrm{s}} = h(x^{\mathrm{s}}, u^{\mathrm{s}})}. 
	\end{align}	
The best possible behavior we can achieve in closed loop is to steer the system to the \textit{optimal reachable equilibrium} which is defined as the minimizer $y^{\mathrm{sr}}$ of
\begin{align}
	\label{eq:costFunc_nonlin}
	J^* &:= \min_{y^\mathrm{s} \in \Zs_y} J(y^\mathrm{s}) := \min_{y^\mathrm{s} \in \Zs_y} \norm{y^\mathrm{s} - y^r}^2_S.
\end{align}
We denote the corresponding steady-state, which we assume is unique (cf.\ Assumption~\ref{ass:UniqueSS} below), by $x^{\mathrm{sr}}$.
 Since the proposed approach relies on least-squares estimation of the linearized dynamics, we also define the corresponding elements for the identified dynamics.
 To be precise, for data $\D_t$, we define the steady-state manifold for the identified system as
 \begin{align*}
 	\hat\Zs \lr{\D_t}  &:= \set{(x^{\mathrm{s}}, u^{\mathrm{s}}) \in \R^n \times \U^{\mathrm{s}}}{x^{\mathrm{s}} = \hat{A}_t x^{\mathrm{s}} + \hat{B}_t u^{\mathrm{s}} + \hat{e}_t},\\
 	\ZsAPPRy{t} &:=  \set{y^{\mathrm{s}} \in \R^p}{\exists (x^{\mathrm{s}}, u^{\mathrm{s}}) \in \hat{\mathcal{Z}}_{\mathrm{y}}^{\mathrm{s}} \lr{\D_t}: \, y^{\mathrm{s}} = \hat{C}_t x^{\mathrm{s}} + \hat{D}_t u^{\mathrm{s}} + \hat r_{t}}.
 \end{align*}
 Moreover, we also define the optimal reachable equilibrium w.r.t. the identified dynamics as the minimizer $\hat y^{\mathrm{sr}}(\D_t)$ of
 \begin{align}
	\costFuncApprOPT{t} &:= \min_{y^\mathrm{s} \in \ZsAPPRy{t}}  J(y^{\mathrm{s}})= \min_{y^\mathrm{s} \in \ZsAPPRy{t}} \norm{y^{\mathrm{s}} - y^{\mathrm{r}}}^2_S. \label{eq:costFunc_ident}
\end{align}
To ensure the existence and uniqueness of the optimal reachable equilibrium, we make following assumption.
\begin{assumption}\label{ass:UniqueSS}
	There exists $\sigma_s > 0$, such that for any $t \geq N$, it holds that $\sigma_{\min}\lr{\begin{bmatrix}
					\hat A_t - I & \hat B_t \\ \hat C_t & \hat D_t
			\end{bmatrix}} \geq \sigma_s$ and the matrix $\begin{bmatrix}
					\hat A_t - I & \hat B_t \\ \hat C_t & \hat D_t
			\end{bmatrix}$ has full column rank.
\end{assumption}
Assumption~\ref{ass:UniqueSS} is a common assumption in tracking MPC~\cite[Remark 1]{limon2018nonlinear}, \cite[Lemma 1.8]{rawlings2017model}.
It means that, for any output steady-state for the identified dynamics, there exists a unique state-input pair which is also a corresponding steady-state.
More precisely, there exist a unique affine map $g_{\Aidx{t}}: \ZsAPPRy{t} \to \hat \Zs\lr{\D_t}$ such that $\hat y^{\mathrm{s}} \mapsto (\hat x^{\mathrm{s}}, \hat u^{\mathrm{s}})$, where $\hat y^{\mathrm{s}} = \hat h_{t}(\hat x^{\mathrm{s}}, \hat u^{\mathrm{s}})$.
	Further, $g_{\Aidx{t}}$ is uniformly Lipschitz continuous for all $ t \geq N$.
We write $(\hat x^{\mathrm{sr}}(\D_t), \hat u^{\mathrm{sr}}(\D_t)) := g_{\Aidx{t}}(\hat y^{\mathrm{sr}}(\D_t))$ for the input and state corresponding to the optimal reachable output $\hat{y}^{\mathrm{sr}}$.
Moreover, we make the following technical assumptions.
\begin{assumption}\label{ass:technical_assumptions}
	\begin{enumerate}
		\item[(a)] The pair $(\hat{A}_{t}, \hat{B}_{t})$ is uniformly controllable for all $t \geq N$, i.e., the minimum singular value of $\begin{bmatrix}
			\hat{B}_t & \dots & \hat{A}_t^{n-1} \hat{B}_{t}
		\end{bmatrix}$ is uniformly bounded. 

		\item[(b)] There exists $\sigma_l > 0$, such that $\sigma_{\min}(I - \hat{A}_t) \geq \sigma_l$ for all $t \geq N$.
		
		\item[(c)] There exists a compact set $\mathcal{B}$ such that $\ZsAPPR{t} \subseteq \mathcal{B}$ for all $t$. 
	\end{enumerate}
\end{assumption} 

Assumption~\ref{ass:technical_assumptions} (a) is required to construct the candidate solution in the proof of our main result.
Moreover, Assumption~\ref{ass:technical_assumptions} (b) implies non-singularity of the dynamics in the sense that, for any steady-state input there exists a unique corresponding state component.
Finally, Assumption~\ref{ass:technical_assumptions} (c) ensures compactness of the steady-state manifold.
We refer to~\cite[Section II.B]{berberich2022linear1} for a detailed discussion of these assumptions and possibilities to satisfy or relax them.

\begin{assumption} \label{ass:lastAss}
	There exist constants $c_l$, $c_u > 0$ such that, for any $t\in\mathbb{N}$, it holds that 
	\begin{align}
		c_l \norm{x_t - \hat x^{\mathrm{sr}}(\D_t)}_2^2 \leq \norm{x_t -  x^{\mathrm{sr}}}_2^2 \leq c_u \norm{x_t - \hat x^{\mathrm{sr}}(\D_t)}_2^2.
	\end{align}
\end{assumption}

Assumption~\ref{ass:lastAss} relates the optimal reachable steady-states of the nonlinear and the identified dynamics.
In~\cite{berberich2022linear1}, it is shown that the assumption holds when replacing $\hat x^{\mathrm{sr}}(\D_t)$ by the optimal reachable steady-state for the linearized dynamics~\eqref{eq:linSys}, assuming that $y^{\mathrm{r}}$ is reachable and $m=p$.
The latter distinction between linearized and identified dynamics does not pose problems for our purposes since the following theoretical analysis ensures in closed loop that the model identified at time $t$ is close to the true linearization at $x_t$ (modulo a small approximation error that can be handled using a robustness argument).

Finally, we assume persistence of excitation (PE) of the closed-loop trajectories generated via Algorithm~\ref{alg:MPC}.
\begin{assumption} \label{ass:PE}
	There exists $\sigma_Z > 0$ such that $\sigma_{\min}(Z_t) \geq \sigma_Z$ and $Z_t$ has full row rank for all $t\in\mathbb{N}$. 
\end{assumption}
While all previous assumptions are satisfied under realistic conditions~\cite{berberich2022linear1}, Assumption~\ref{ass:PE} is not guaranteed a priori under the proposed control approach.
Although a complete theoretical solution to ensuring closed-loop PE remains an open problem, there are various effective heuristics that yield good practical results, e.g.,
adding a PE disturbance signal to the input, stopping the data updates once the system is close to the setpoint, or adding a cost term which encourages PE trajectories~\cite{lu2021robust}.

\subsection{Practical Exponential Stability}\label{subsec:stability}

The following result shows that, under the above assumptions, the proposed MPC scheme practically exponentially stabilizes the optimal reachable steady-state $x^{\mathrm{s}\mathrm{r}}$ for the nonlinear system~\eqref{eq:nonlinSys} in closed loop.
The proof relies on the Lyapunov function candidate 
\begin{align}\label{eq:lyap_function_candidate}
    V(x_t,\D_t)\coloneqq \costFuncMpcOPT{t}-\hat{J}^*(\D_t).
\end{align}

\begin{theorem} \label{thm:mainResult}
	Suppose $t \geq N \geq n$ and Assumptions \ref{ass:UniqueSS} -- \ref{ass:PE} hold. Then, there exists $ \bar{\varepsilon}>0$ such that, for any $\varepsilon\in[0,\bar{\varepsilon}]$, there exist $V_{\max}, s_u , C>0 $ and $0 < c_V < 1$, $\beta \in \mathcal{K}_{\infty}$, such that, if
	\begin{align}\label{eq:thm_stability_conditions_S_V}
		& \lambda_{\max} (S) < s_u, V(x_N,\D_N) \leq V_{\max}
		\\
		&\norm{x_N - x_k}_2^2 \leq  \varepsilon \text{ for all } k = 0, \dots, N-1, \label{eq:mp:mainThmError}
	\end{align}
	then, for any $ t= N + ni$, $i \geq 0$, the problem \eqref{eq:myMPC} is feasible and the closed loop under Algorithm \ref{alg:MPC} satisfies 
	\begin{align}
		\norm{x_t - x^{\mathrm{sr}}}_2^2 \leq c_V^i C \norm{x_0 -  x^{\mathrm{sr}}}_2^2 + \beta(\varepsilon). \label{eq:mp:mainResult}
	\end{align}
\end{theorem}
\begin{proof}
The proof follows similar arguments as~\cite[Theorem 1]{berberich2022linear1} and~\cite[Theorem 2]{berberich2022linear2}, and we only provide a sketch in the following.
For simplicity of exposition, we assume that all state values considered in this proof lie within a compact set $X$.
This can be ensured using an invariance argument based on the employed Lyapunov function, compare the proofs of~\cite[Theorem 1]{berberich2022linear1}, \cite[Theorem 2]{berberich2022linear2} for details.

\textbf{1) System identification bound}\\
The main difference to~\cite[Theorem 1]{berberich2022linear1} is that the proposed MPC scheme relies on the identified dynamics~\eqref{eq:identSys} rather than the linearized dynamics~\eqref{eq:linSys}.
Therefore, we first derive a bound on the difference between the two.
To this end, we write the nonlinear dynamics applied to the past data points $k=t-N,\dots,t-1$ as
\begin{align}\label{eq:thm_proof_sysid_linearized_dynamics}
    f(x_k,u_k)=f_{x_t}(x_k,u_k)+\Delta_{t,k}
\end{align}
for some $\Delta_{t,k}$ which can be bounded using a Taylor series argument as
\begin{align}\label{eq:thm_proof_sysid_delta_bound}
    \|\Delta_{t,k}\rVert_2\leq c_{\Delta,1}\| x_t-x_k\rVert_2^2
\end{align}
for some $c_{\Delta,1}>0$, see~\cite[Lemma 1]{berberich2022linear2}.
In view of~\eqref{eq:thm_proof_sysid_linearized_dynamics}, we can interpret the process of estimating $f_{x_t}$ from data of the nonlinear system $f$ as a \emph{perturbed} identification problem with nominal affine dynamics $f_{x_t}$ and disturbance $\Delta_{t,k}$.
This allows us to use standard error bounds for least-squares estimation~\cite{ziemann2023tutorial}:
Defining  $\Delta_{[t]}:=\begin{bmatrix}
        \Delta_{t,t-N}&\Delta_{t,t-N+1}&\dots&\Delta_{t,t-1}
    \end{bmatrix}$,
we obtain 
\begin{align}
    &\begin{bmatrix}
        \hat{A}_t&\hat{B}_t&\hat{e}_t
    \end{bmatrix}
    -\begin{bmatrix}
        A_{x_t}&B&e_{x_t}
    \end{bmatrix}\\\nonumber 
    =&X_t^+Z_t^\top(Z_tZ_t^\top)^{-1}
    -(X_t^+-\Delta_{[t]})(Z_t^\top (Z_tZ_t^\top)^{-1}
    =\Delta_{[t]}Z_t^\top (Z_tZ_t^\top)^{-1}.
\end{align}
Thus, using~\eqref{eq:thm_proof_sysid_delta_bound}, the PE bound in Assumption~\ref{ass:PE}, as well as compactness of $X$ and $\U$, we obtain 
\begin{align}
		\|\hat f_{t}(x, u) - f_{x_t}(x, u)\rVert_2 &\leq c_{\Delta,2}\| \Delta_{[t]}\rVert_2
  \leq c_{\Delta,3}\sum_{k=t-N}^{t-1}\| x_t-x_k\rVert_2^2
  \label{eq:mp:varepsilon_f}
	\end{align}
for any $x\in X$, $u\in\U$ and with some $c_{\Delta,2},c_{\Delta,3}>0$.
Analogously, it can be shown that
\begin{align}
		\|\hat h_{t}(x, u) - h_{x_t}(x, u)\rVert_2 
  \leq c_{\Delta,4}\sum_{k=t-N}^{t-1}\| x_t-x_k\rVert_2^2
  \label{eq:mp:varepsilon_h}
	\end{align}
for any $x\in X$, $u\in\U$ and with some $c_{\Delta,4}>0$.

\textbf{2) Lyapunov function bounds}\\
Suppose now that the bound~\eqref{eq:mp:mainThmError} holds and let $t=N$.
This implies that the error between the linearized dynamics~\eqref{eq:linSys} and the identified dynamics~\eqref{eq:identSys} is bounded as
\begin{align}\label{eq:thm_proof_lyapunov_id_bound}
    \|\hat f_{t}(x, u) - f_{x_t}(x, u)\rVert_2&\leq c_f\varepsilon,\\\nonumber 
    \|\hat h_{t}(x, u) - h_{x_t}(x, u)\rVert_2&\leq c_h\varepsilon
\end{align}
for some $c_f,c_h>0$, compare~\eqref{eq:mp:varepsilon_f} and~\eqref{eq:mp:varepsilon_h}.
Using an inherent robustness argument to cope with the identification error~\eqref{eq:thm_proof_lyapunov_id_bound}, compare~\cite{berberich2022inherent}, a straightforward adaptation of the proof of~\cite[Theorem 1]{berberich2022linear1} yields the following:
If $\varepsilon$ is sufficiently small, $\LyapFunc{t} \leq V_{\max}$, $\hat J^*(\D_t) \leq J_{\max}$, $\hat J^*(\D_{t+n}^{\top}) \leq J_{\max}$ for some $J_{\max}>0$, then the problem \eqref{eq:myMPC} is feasible at time $t+n$ and there exist $0 < c_V < 1$, $c_\varepsilon > 0$, such that 
\begin{align}
    \LyapFunc{t+n} \leq c_V \LyapFunc{t} + c_\varepsilon \varepsilon\sqrt{V_{\max}+J_{\max}}. \label{eq:thm_proof_lyapunov_decay_bound}
\end{align}
Using similar ideas as in the proof of~\cite[Theorem 2]{berberich2022linear2}, it can be shown that a variation of the bound~\eqref{eq:mp:mainThmError} holds recursively.
To be precise, if
\begin{align}
    \sum_{k=j}^{j+n}\| x_t-x_k\rVert_2^2\leq (V_{\max}+J_{\max})^{\alpha}
\end{align}
for $j=t-N,\dots,t-n$, then
\begin{align}
    \sum_{k=t}^{t+n}\| x_t-x_k\rVert_2^2\leq (V_{\max}+J_{\max})^{\alpha}
\end{align}
with some $\alpha\in(0,1)$.
As a result, the bound~\eqref{eq:mp:mainThmError} holds recursively as well in the sense that 
\begin{align}
    \| x_{t+n}-x_k\rVert_2^2\leq c_{\varepsilon}(V_{\max}+J_{\max})^{\alpha}
\end{align}
for some $c_{\varepsilon}>0$ and for $k=t+n-N,\dots,t+n-1$.
Choose now $\varepsilon=(V_{\max}+J_{\max})^{\alpha}$ and note that $J_{\max}$ can be made arbitrarily small when $\lambda_{\max}(S)$ is sufficiently small.
Then,~\eqref{eq:thm_proof_lyapunov_decay_bound} implies 
\begin{align}
    \LyapFunc{t+n} \leq c_V \LyapFunc{t} + c_\varepsilon (V_{\max}+J_{\max})^{\alpha+1/2}.
\end{align}
Moreover, choose $J_{\max} = V_{\max}$, and $\alpha > 1/2$.
Then, if $V_{\max}$ is sufficiently small, we infer $\LyapFunc{t+n}\leq V_{\max}$ such that the above bounds hold for $t=N+n$ and, therefore, they hold recursively for all $t=N+ni$, $i\in\mathbb{N}$.

Finally, in analogy to~\cite[Lemma 1]{berberich2022linear1}, suitable lower and upper bounds on the Lyapunov function $\LyapFunc{t}$ can be obtained.
It now follows from standard Lyapunov arguments together with Assumption~\ref{ass:lastAss} that $x^{\mathrm{s}\mathrm{r}}$ is practically exponentially stable.
  \hfill \qed
\end{proof}

Theorem~\ref{thm:mainResult} shows that the optimal reachable steady-state $x^{\mathrm{s}\mathrm{r}}$ of the nonlinear system~\eqref{eq:nonlinSys} is practically exponentially stabilized in closed loop.
This means that the closed-loop trajectory exponentially converges to a region around $x^{\mathrm{s}\mathrm{r}}$ whose size increases with the bound $\varepsilon$ on the distance between the initial data points, compare~\eqref{eq:mp:mainThmError}.
The main proof idea is that, if the latter distance is small initially, then it remains small in closed loop such that the estimated linear model serves as a good local approximation of the nonlinear dynamics.
Beyond the technical assumptions from Section~\ref{subsec:assumptions}, the result requires~\eqref{eq:thm_stability_conditions_S_V}, meaning that the matrix $S$ is sufficiently small and that the Lyapunov function candidate $V(x_N,\D_N)$ is sufficiently small, which can be fulfilled if the initial state $x_0$ is sufficiently close to the steady-state manifold.

Theorem~\ref{thm:mainResult} should be viewed as a qualitative theoretical result, providing insights into the interplay of different system, design, and data parameters as well as the associated closed-loop guarantees.
Rigorously verifying all the associated conditions, e.g., computing bounds $s_u$, $V_{\max}$ as required in~\eqref{eq:thm_stability_conditions_S_V}, is challenging in general and poses an interesting direction for future research.

\section{Numerical Example}\label{sec:example}

In this chapter, we present numerical simulations for the proposed MPC algorithm and compare it to the ones presented in \cite{berberich2022linear1} and \cite{berberich2022linear2}.
The code for the following results is available on \url{https://github.com/tastr/IdentificationBasedMPC}.

The proposed MPC scheme is applied to the nonlinear continuous stirred tank reactor from~\cite{mayne2011tube}, which was also considered in \cite{berberich2022linear1} and \cite{berberich2022linear2} and is given by
\begin{align} \label{eq:ExpSystem}
	f(x,u) :=& \begin{pmatrix}
		x_1 + \frac{T_s}{\theta} (1 - x_1) - T_s \bar k e^{-\frac{M}{x_2}} 
		\\
		x_2 + \frac{T_s}{\theta} (x_f - x_2) + T_s \bar k x_1 e^{-\frac{M}{x_2}} - T_s \alpha u(x_2 - x_c)
	\end{pmatrix}, \\
	h(x,u) :=& x_2.
\end{align}
The states $x_1$ and $x_2 = y$ represent temperature and concentration.
The input $u$ is the coolant flow rate. 
The system results from a Euler discretization of the continuous-time dynamics with time-step $T_s = 0.2$. 
The values of the other parameters are $\theta = 20$, $\bar k = 300$, $M = 5$, $x_f = 0.3947$, $x_c = 0.3816$, $\alpha = 0.117$. 
We choose the input constraint sets $ \U = [0.1, 2]$, $\U^\mathrm{s} = [0.11, 1.99]$, the weighting matrices $Q = I$ and $R = 0.05$, and the prediction horizon $L = 41$.
As the above system is not input-affine, we define a new input $\Delta u_k := u_{k+1} - u_k$ and treat $u$ as an additional state. 
The control goal is to steer the output to the setpoint $y^{\mathrm{r}} = 0.6519$ when starting from the initial condition $x_0 = \begin{bmatrix}0.4&0.6&0.1\end{bmatrix}^{\top}$.
For better practical reliability, we add a regularization to the identification approach introduced in Section~\ref{subsec:system_identification}, i.e., instead of~\eqref{eq:lsq_explicit_solution} we compute the estimated quantities as 
\begin{align}\label{eq:lsq_explicit_solution_regularized}
	\begin{bmatrix}
		\hat{A}_t & \hat{B}_t & \hat{e}_t
	\end{bmatrix} &= X_t^+ Z_t^{\top} (Z_t Z_t^{\top}+\lambda I)^{-1},\>\>
	\begin{bmatrix}
		\hat{C}_t & \hat{D}_t & \hat{r}_t
	\end{bmatrix} = Y_t Z_t^{\top} (Z_t Z_t^{\top}+\lambda I)^{-1}
\end{align}
with regularization parameter $\lambda=10^{-12}$.
Extending our theoretical results to such a regularized least-squares estimate is an interesting direction for future research.
As discussed in \cite{berberich2022linear1}, Assumptions~\ref{ass:UniqueSS}--\ref{ass:lastAss} are fulfilled for the above system.
In order to achieve closed-loop PE (Assumption \ref{ass:PE}),  we stop the identification update once the state difference between two time steps is smaller than $5\cdot10^{-6}$. 

Figure \ref{fig:solExample} shows the simulation results for the proposed indirect data-driven MPC scheme as well as the ones from~\cite{berberich2022linear1} (model-based) and~\cite{berberich2022linear2} (direct data-driven).
For the model-based and the indirect data-driven approach, the weight matrix $S$ is set to $S=100$, whereas for the direct data-driven approach it is $S = 10$. The reason is that the latter includes a slack variable which encourages faster convergence, hence necessitating a smaller choice of $S$ to avoid too fast changes in the state evolution.
Further, we note that, for the direct and indirect data-driven approaches, the initial data set was generated based on the model-based approach in order to simplify the comparison.

While the closed loop converges to the desired value for all three algorithms, we observe that the smoothest trajectory is generated by the model-based approach, followed by the identification-based. The direct data-driven approach shows slightly slower convergence (which is partly due to the smaller value for $S$) as well as overshooting behavior. Moreover, the data length $N$ needed for successful tracking in the direct data-driven MPC scheme from~\cite{berberich2022linear2} is substantially larger than in the present indirect approach.
To be precise, while we use $N=25$ data points for the proposed approach, the one from~\cite{berberich2022linear2} requires a theoretical lower bound of $N=90$ data points and, for the shown simulations, uses $N=120$ data points.

\begin{figure}
\begin{center}
	\includegraphics[width=0.9\textwidth]{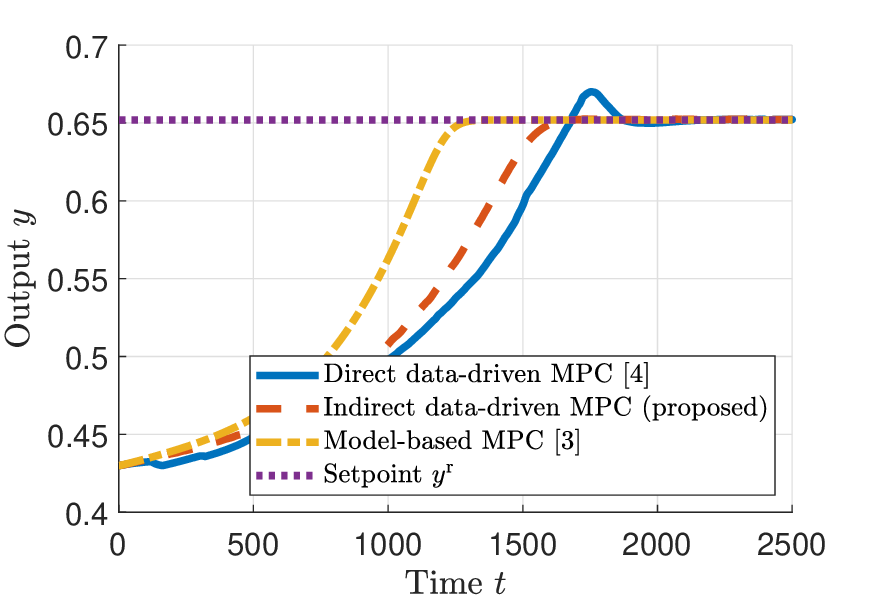}
 \end{center}
	\caption{Closed-loop simulations under the MPC schemes from~\cite{berberich2022linear1} (model-based MPC) and~\cite{berberich2022linear2} (direct data-driven MPC) as well as the approach proposed in the present paper (indirect data-driven MPC).}
	\label{fig:solExample}
\end{figure}

%Further, we compare the hyper-parameter sensitivity for identification-based and data-based algorithms,
Further, we compare the sensitivity to hyper-parameters for the indirect and direct data-driven MPC schemes by comparing the tracking error 
\begin{align}\label{eq:error}
\sum_{t=0}^{2500}\norm{y_t - y^{\mathrm{r}}}_2
\end{align} 
for different parameters. The Figures \ref{fig:LSQMatrix_indirect} and \ref{fig:LSQMatrix_direct} show the dependency of the error on the regularization parameters and the data length $N$ for both approaches. 
The white cells denote simulations with numerical issues (in particular, infeasibility of the optimization).
It can be seen that the indirect approach has significantly improved robustness against parameter variations compared to the direct approach.
In particular, while the direct data-driven MPC scheme from~\cite{berberich2022linear2} requires a data length (roughly) below $N=150$ to avoid a too high linearization error, the proposed indirect approach can cope with substantially higher values up to $N=300$.
Similarly, the range of the regularization parameters leading to satisfactory performance is significantly higher in the indirect approach compared to the direct one.

\section{Conclusion}\label{sec:conclusion}
In this paper, we have introduced an indirect data-driven MPC scheme for nonlinear systems combining online linear system identification and tracking MPC, and admitting desirable closed-loop guarantees.
The presented indirect approach is closely related to the existing direct data-driven MPC scheme from~\cite{berberich2022linear2} which uses the Fundamental Lemma for prediction.
Our analysis reveals that, despite the fundamental difference between the two employed prediction models (identified model vs.\ Fundamental Lemma), both schemes admit comparable theoretical guarantees.
Yet, several notable differences should be highlighted.

\begin{figure}[H]
	\centering
	\input{Plots/lsqMatrix.tex}
	\caption{Error~\eqref{eq:error} for indirect data-driven MPC (proposed approach).}
	\label{fig:LSQMatrix_indirect}
	\centering
 
	\vskip17pt
 \input{Plots/ddMatrix.tex}
	\caption{
 Error~\eqref{eq:error} for direct data-driven MPC (approach from~\cite{berberich2022linear2}).}
	\label{fig:LSQMatrix_direct}
\end{figure}

In comparison to~\cite{berberich2022linear2}, the main benefit of the indirect data-driven MPC scheme is that it is simpler to tune and admits a (slightly) easier theoretical analysis.
Further, contrary to its direct counterpart, the computational complexity of the indirect approach does not scale with the data length.
While the complexity of the employed identification algorithm scales with the data length, a recursive identification algorithm could be used to overcome this issue.
On the other hand, the main benefit of the direct approach from~\cite{berberich2022linear2} is that it allows to control unknown system based only on input-output data, whereas the presented indirect scheme requires state measurements.
Finally, the indirect approach employs a one-step model whereas the direct approach uses a multi-step predictor, which causes further inherent differences~\cite{koehler2022state}.
A more formal theoretical investigation of the connections between direct and indirect data-driven MPC is an interesting issue for future research.

\section{Acknowledgment}
Johannes K\"ohler was supported by the Swiss National Science Foundation under NCCR Automation (grant agreement 51NF40 180545).

 \bibliographystyle{plain}
 \bibliography{lit}
\end{document}

%% file: Plots/lsqMatrix.tex
\begin{tikzpicture}

\begin{axis}[%
width=.75\textwidth,
height=4cm,
at={(0.125in,0.105in)},
scale only axis,
point meta min=6.5904673443022,
point meta max=11.3117622206068,
xmin=1,
xmax=7,
xtick={1,2,3,4,5,6,7},
xticklabels={{0},{$10^{-12}$},{$10^{-11}$},{$10^{-10}$},{$10^{-9}$},{$10^{-8}$},{$10^{-7}$}},
xlabel style={font=\color{white!15!black}},
xlabel={Regularization parameter $\lambda$},
ymin=30,
ymax=300,
ylabel style={font=\color{white!15!black}},
ylabel={Data length $N$},
axis background/.style={fill=white},
colormap={mymap}{[1pt] rgb(0pt)=(0.2422,0.1504,0.6603); rgb(1pt)=(0.2444,0.1534,0.6728); rgb(2pt)=(0.2464,0.1569,0.6847); rgb(3pt)=(0.2484,0.1607,0.6961); rgb(4pt)=(0.2503,0.1648,0.7071); rgb(5pt)=(0.2522,0.1689,0.7179); rgb(6pt)=(0.254,0.1732,0.7286); rgb(7pt)=(0.2558,0.1773,0.7393); rgb(8pt)=(0.2576,0.1814,0.7501); rgb(9pt)=(0.2594,0.1854,0.761); rgb(11pt)=(0.2628,0.1932,0.7828); rgb(12pt)=(0.2645,0.1972,0.7937); rgb(13pt)=(0.2661,0.2011,0.8043); rgb(14pt)=(0.2676,0.2052,0.8148); rgb(15pt)=(0.2691,0.2094,0.8249); rgb(16pt)=(0.2704,0.2138,0.8346); rgb(17pt)=(0.2717,0.2184,0.8439); rgb(18pt)=(0.2729,0.2231,0.8528); rgb(19pt)=(0.274,0.228,0.8612); rgb(20pt)=(0.2749,0.233,0.8692); rgb(21pt)=(0.2758,0.2382,0.8767); rgb(22pt)=(0.2766,0.2435,0.884); rgb(23pt)=(0.2774,0.2489,0.8908); rgb(24pt)=(0.2781,0.2543,0.8973); rgb(25pt)=(0.2788,0.2598,0.9035); rgb(26pt)=(0.2794,0.2653,0.9094); rgb(27pt)=(0.2798,0.2708,0.915); rgb(28pt)=(0.2802,0.2764,0.9204); rgb(29pt)=(0.2806,0.2819,0.9255); rgb(30pt)=(0.2809,0.2875,0.9305); rgb(31pt)=(0.2811,0.293,0.9352); rgb(32pt)=(0.2813,0.2985,0.9397); rgb(33pt)=(0.2814,0.304,0.9441); rgb(34pt)=(0.2814,0.3095,0.9483); rgb(35pt)=(0.2813,0.315,0.9524); rgb(36pt)=(0.2811,0.3204,0.9563); rgb(37pt)=(0.2809,0.3259,0.96); rgb(38pt)=(0.2807,0.3313,0.9636); rgb(39pt)=(0.2803,0.3367,0.967); rgb(40pt)=(0.2798,0.3421,0.9702); rgb(41pt)=(0.2791,0.3475,0.9733); rgb(42pt)=(0.2784,0.3529,0.9763); rgb(43pt)=(0.2776,0.3583,0.9791); rgb(44pt)=(0.2766,0.3638,0.9817); rgb(45pt)=(0.2754,0.3693,0.984); rgb(46pt)=(0.2741,0.3748,0.9862); rgb(47pt)=(0.2726,0.3804,0.9881); rgb(48pt)=(0.271,0.386,0.9898); rgb(49pt)=(0.2691,0.3916,0.9912); rgb(50pt)=(0.267,0.3973,0.9924); rgb(51pt)=(0.2647,0.403,0.9935); rgb(52pt)=(0.2621,0.4088,0.9946); rgb(53pt)=(0.2591,0.4145,0.9955); rgb(54pt)=(0.2556,0.4203,0.9965); rgb(55pt)=(0.2517,0.4261,0.9974); rgb(56pt)=(0.2473,0.4319,0.9983); rgb(57pt)=(0.2424,0.4378,0.9991); rgb(58pt)=(0.2369,0.4437,0.9996); rgb(59pt)=(0.2311,0.4497,0.9995); rgb(60pt)=(0.225,0.4559,0.9985); rgb(61pt)=(0.2189,0.462,0.9968); rgb(62pt)=(0.2128,0.4682,0.9948); rgb(63pt)=(0.2066,0.4743,0.9926); rgb(64pt)=(0.2006,0.4803,0.9906); rgb(65pt)=(0.195,0.4861,0.9887); rgb(66pt)=(0.1903,0.4919,0.9867); rgb(67pt)=(0.1869,0.4975,0.9844); rgb(68pt)=(0.1847,0.503,0.9819); rgb(69pt)=(0.1831,0.5084,0.9793); rgb(70pt)=(0.1818,0.5138,0.9766); rgb(71pt)=(0.1806,0.5191,0.9738); rgb(72pt)=(0.1795,0.5244,0.9709); rgb(73pt)=(0.1785,0.5296,0.9677); rgb(74pt)=(0.1778,0.5349,0.9641); rgb(75pt)=(0.1773,0.5401,0.9602); rgb(76pt)=(0.1768,0.5452,0.956); rgb(77pt)=(0.1764,0.5504,0.9516); rgb(78pt)=(0.1755,0.5554,0.9473); rgb(79pt)=(0.174,0.5605,0.9432); rgb(80pt)=(0.1716,0.5655,0.9393); rgb(81pt)=(0.1686,0.5705,0.9357); rgb(82pt)=(0.1649,0.5755,0.9323); rgb(83pt)=(0.161,0.5805,0.9289); rgb(84pt)=(0.1573,0.5854,0.9254); rgb(85pt)=(0.154,0.5902,0.9218); rgb(86pt)=(0.1513,0.595,0.9182); rgb(87pt)=(0.1492,0.5997,0.9147); rgb(88pt)=(0.1475,0.6043,0.9113); rgb(89pt)=(0.1461,0.6089,0.908); rgb(90pt)=(0.1446,0.6135,0.905); rgb(91pt)=(0.1429,0.618,0.9022); rgb(92pt)=(0.1408,0.6226,0.8998); rgb(93pt)=(0.1383,0.6272,0.8975); rgb(94pt)=(0.1354,0.6317,0.8953); rgb(95pt)=(0.1321,0.6363,0.8932); rgb(96pt)=(0.1288,0.6408,0.891); rgb(97pt)=(0.1253,0.6453,0.8887); rgb(98pt)=(0.1219,0.6497,0.8862); rgb(99pt)=(0.1185,0.6541,0.8834); rgb(100pt)=(0.1152,0.6584,0.8804); rgb(101pt)=(0.1119,0.6627,0.877); rgb(102pt)=(0.1085,0.6669,0.8734); rgb(103pt)=(0.1048,0.671,0.8695); rgb(104pt)=(0.1009,0.675,0.8653); rgb(105pt)=(0.0964,0.6789,0.8609); rgb(106pt)=(0.0914,0.6828,0.8562); rgb(107pt)=(0.0855,0.6865,0.8513); rgb(108pt)=(0.0789,0.6902,0.8462); rgb(109pt)=(0.0713,0.6938,0.8409); rgb(110pt)=(0.0628,0.6972,0.8355); rgb(111pt)=(0.0535,0.7006,0.8299); rgb(112pt)=(0.0433,0.7039,0.8242); rgb(113pt)=(0.0328,0.7071,0.8183); rgb(114pt)=(0.0234,0.7103,0.8124); rgb(115pt)=(0.0155,0.7133,0.8064); rgb(116pt)=(0.0091,0.7163,0.8003); rgb(117pt)=(0.0046,0.7192,0.7941); rgb(118pt)=(0.0019,0.722,0.7878); rgb(119pt)=(0.0009,0.7248,0.7815); rgb(120pt)=(0.0018,0.7275,0.7752); rgb(121pt)=(0.0046,0.7301,0.7688); rgb(122pt)=(0.0094,0.7327,0.7623); rgb(123pt)=(0.0162,0.7352,0.7558); rgb(124pt)=(0.0253,0.7376,0.7492); rgb(125pt)=(0.0369,0.74,0.7426); rgb(126pt)=(0.0504,0.7423,0.7359); rgb(127pt)=(0.0638,0.7446,0.7292); rgb(128pt)=(0.077,0.7468,0.7224); rgb(129pt)=(0.0899,0.7489,0.7156); rgb(130pt)=(0.1023,0.751,0.7088); rgb(131pt)=(0.1141,0.7531,0.7019); rgb(132pt)=(0.1252,0.7552,0.695); rgb(133pt)=(0.1354,0.7572,0.6881); rgb(134pt)=(0.1448,0.7593,0.6812); rgb(135pt)=(0.1532,0.7614,0.6741); rgb(136pt)=(0.1609,0.7635,0.6671); rgb(137pt)=(0.1678,0.7656,0.6599); rgb(138pt)=(0.1741,0.7678,0.6527); rgb(139pt)=(0.1799,0.7699,0.6454); rgb(140pt)=(0.1853,0.7721,0.6379); rgb(141pt)=(0.1905,0.7743,0.6303); rgb(142pt)=(0.1954,0.7765,0.6225); rgb(143pt)=(0.2003,0.7787,0.6146); rgb(144pt)=(0.2061,0.7808,0.6065); rgb(145pt)=(0.2118,0.7828,0.5983); rgb(146pt)=(0.2178,0.7849,0.5899); rgb(147pt)=(0.2244,0.7869,0.5813); rgb(148pt)=(0.2318,0.7887,0.5725); rgb(149pt)=(0.2401,0.7905,0.5636); rgb(150pt)=(0.2491,0.7922,0.5546); rgb(151pt)=(0.2589,0.7937,0.5454); rgb(152pt)=(0.2695,0.7951,0.536); rgb(153pt)=(0.2809,0.7964,0.5266); rgb(154pt)=(0.2929,0.7975,0.517); rgb(155pt)=(0.3052,0.7985,0.5074); rgb(156pt)=(0.3176,0.7994,0.4975); rgb(157pt)=(0.3301,0.8002,0.4876); rgb(158pt)=(0.3424,0.8009,0.4774); rgb(159pt)=(0.3548,0.8016,0.4669); rgb(160pt)=(0.3671,0.8021,0.4563); rgb(161pt)=(0.3795,0.8026,0.4454); rgb(162pt)=(0.3921,0.8029,0.4344); rgb(163pt)=(0.405,0.8031,0.4233); rgb(164pt)=(0.4184,0.803,0.4122); rgb(165pt)=(0.4322,0.8028,0.4013); rgb(166pt)=(0.4463,0.8024,0.3904); rgb(167pt)=(0.4608,0.8018,0.3797); rgb(168pt)=(0.4753,0.8011,0.3691); rgb(169pt)=(0.4899,0.8002,0.3586); rgb(170pt)=(0.5044,0.7993,0.348); rgb(171pt)=(0.5187,0.7982,0.3374); rgb(172pt)=(0.5329,0.797,0.3267); rgb(173pt)=(0.547,0.7957,0.3159); rgb(175pt)=(0.5748,0.7929,0.2941); rgb(176pt)=(0.5886,0.7913,0.2833); rgb(177pt)=(0.6024,0.7896,0.2726); rgb(178pt)=(0.6161,0.7878,0.2622); rgb(179pt)=(0.6297,0.7859,0.2521); rgb(180pt)=(0.6433,0.7839,0.2423); rgb(181pt)=(0.6567,0.7818,0.2329); rgb(182pt)=(0.6701,0.7796,0.2239); rgb(183pt)=(0.6833,0.7773,0.2155); rgb(184pt)=(0.6963,0.775,0.2075); rgb(185pt)=(0.7091,0.7727,0.1998); rgb(186pt)=(0.7218,0.7703,0.1924); rgb(187pt)=(0.7344,0.7679,0.1852); rgb(188pt)=(0.7468,0.7654,0.1782); rgb(189pt)=(0.759,0.7629,0.1717); rgb(190pt)=(0.771,0.7604,0.1658); rgb(191pt)=(0.7829,0.7579,0.1608); rgb(192pt)=(0.7945,0.7554,0.157); rgb(193pt)=(0.806,0.7529,0.1546); rgb(194pt)=(0.8172,0.7505,0.1535); rgb(195pt)=(0.8281,0.7481,0.1536); rgb(196pt)=(0.8389,0.7457,0.1546); rgb(197pt)=(0.8495,0.7435,0.1564); rgb(198pt)=(0.86,0.7413,0.1587); rgb(199pt)=(0.8703,0.7392,0.1615); rgb(200pt)=(0.8804,0.7372,0.165); rgb(201pt)=(0.8903,0.7353,0.1695); rgb(202pt)=(0.9,0.7336,0.1749); rgb(203pt)=(0.9093,0.7321,0.1815); rgb(204pt)=(0.9184,0.7308,0.189); rgb(205pt)=(0.9272,0.7298,0.1973); rgb(206pt)=(0.9357,0.729,0.2061); rgb(207pt)=(0.944,0.7285,0.2151); rgb(208pt)=(0.9523,0.7284,0.2237); rgb(209pt)=(0.9606,0.7285,0.2312); rgb(210pt)=(0.9689,0.7292,0.2373); rgb(211pt)=(0.977,0.7304,0.2418); rgb(212pt)=(0.9842,0.733,0.2446); rgb(213pt)=(0.99,0.7365,0.2429); rgb(214pt)=(0.9946,0.7407,0.2394); rgb(215pt)=(0.9966,0.7458,0.2351); rgb(216pt)=(0.9971,0.7513,0.2309); rgb(217pt)=(0.9972,0.7569,0.2267); rgb(218pt)=(0.9971,0.7626,0.2224); rgb(219pt)=(0.9969,0.7683,0.2181); rgb(220pt)=(0.9966,0.774,0.2138); rgb(221pt)=(0.9962,0.7798,0.2095); rgb(222pt)=(0.9957,0.7856,0.2053); rgb(223pt)=(0.9949,0.7915,0.2012); rgb(224pt)=(0.9938,0.7974,0.1974); rgb(225pt)=(0.9923,0.8034,0.1939); rgb(226pt)=(0.9906,0.8095,0.1906); rgb(227pt)=(0.9885,0.8156,0.1875); rgb(228pt)=(0.9861,0.8218,0.1846); rgb(229pt)=(0.9835,0.828,0.1817); rgb(230pt)=(0.9807,0.8342,0.1787); rgb(231pt)=(0.9778,0.8404,0.1757); rgb(232pt)=(0.9748,0.8467,0.1726); rgb(233pt)=(0.972,0.8529,0.1695); rgb(234pt)=(0.9694,0.8591,0.1665); rgb(235pt)=(0.9671,0.8654,0.1636); rgb(236pt)=(0.9651,0.8716,0.1608); rgb(237pt)=(0.9634,0.8778,0.1582); rgb(238pt)=(0.9619,0.884,0.1557); rgb(239pt)=(0.9608,0.8902,0.1532); rgb(240pt)=(0.9601,0.8963,0.1507); rgb(241pt)=(0.9596,0.9023,0.148); rgb(242pt)=(0.9595,0.9084,0.145); rgb(243pt)=(0.9597,0.9143,0.1418); rgb(244pt)=(0.9601,0.9203,0.1382); rgb(245pt)=(0.9608,0.9262,0.1344); rgb(246pt)=(0.9618,0.932,0.1304); rgb(247pt)=(0.9629,0.9379,0.1261); rgb(248pt)=(0.9642,0.9437,0.1216); rgb(249pt)=(0.9657,0.9494,0.1168); rgb(250pt)=(0.9674,0.9552,0.1116); rgb(251pt)=(0.9692,0.9609,0.1061); rgb(252pt)=(0.9711,0.9667,0.1001); rgb(253pt)=(0.973,0.9724,0.0938); rgb(254pt)=(0.9749,0.9782,0.0872); rgb(255pt)=(0.9769,0.9839,0.0805)},
colorbar
]

\addplot[%
surf,
shader=flat corner, draw=black, colormap={mymap}{[1pt] rgb(0pt)=(0.2422,0.1504,0.6603); rgb(1pt)=(0.2444,0.1534,0.6728); rgb(2pt)=(0.2464,0.1569,0.6847); rgb(3pt)=(0.2484,0.1607,0.6961); rgb(4pt)=(0.2503,0.1648,0.7071); rgb(5pt)=(0.2522,0.1689,0.7179); rgb(6pt)=(0.254,0.1732,0.7286); rgb(7pt)=(0.2558,0.1773,0.7393); rgb(8pt)=(0.2576,0.1814,0.7501); rgb(9pt)=(0.2594,0.1854,0.761); rgb(11pt)=(0.2628,0.1932,0.7828); rgb(12pt)=(0.2645,0.1972,0.7937); rgb(13pt)=(0.2661,0.2011,0.8043); rgb(14pt)=(0.2676,0.2052,0.8148); rgb(15pt)=(0.2691,0.2094,0.8249); rgb(16pt)=(0.2704,0.2138,0.8346); rgb(17pt)=(0.2717,0.2184,0.8439); rgb(18pt)=(0.2729,0.2231,0.8528); rgb(19pt)=(0.274,0.228,0.8612); rgb(20pt)=(0.2749,0.233,0.8692); rgb(21pt)=(0.2758,0.2382,0.8767); rgb(22pt)=(0.2766,0.2435,0.884); rgb(23pt)=(0.2774,0.2489,0.8908); rgb(24pt)=(0.2781,0.2543,0.8973); rgb(25pt)=(0.2788,0.2598,0.9035); rgb(26pt)=(0.2794,0.2653,0.9094); rgb(27pt)=(0.2798,0.2708,0.915); rgb(28pt)=(0.2802,0.2764,0.9204); rgb(29pt)=(0.2806,0.2819,0.9255); rgb(30pt)=(0.2809,0.2875,0.9305); rgb(31pt)=(0.2811,0.293,0.9352); rgb(32pt)=(0.2813,0.2985,0.9397); rgb(33pt)=(0.2814,0.304,0.9441); rgb(34pt)=(0.2814,0.3095,0.9483); rgb(35pt)=(0.2813,0.315,0.9524); rgb(36pt)=(0.2811,0.3204,0.9563); rgb(37pt)=(0.2809,0.3259,0.96); rgb(38pt)=(0.2807,0.3313,0.9636); rgb(39pt)=(0.2803,0.3367,0.967); rgb(40pt)=(0.2798,0.3421,0.9702); rgb(41pt)=(0.2791,0.3475,0.9733); rgb(42pt)=(0.2784,0.3529,0.9763); rgb(43pt)=(0.2776,0.3583,0.9791); rgb(44pt)=(0.2766,0.3638,0.9817); rgb(45pt)=(0.2754,0.3693,0.984); rgb(46pt)=(0.2741,0.3748,0.9862); rgb(47pt)=(0.2726,0.3804,0.9881); rgb(48pt)=(0.271,0.386,0.9898); rgb(49pt)=(0.2691,0.3916,0.9912); rgb(50pt)=(0.267,0.3973,0.9924); rgb(51pt)=(0.2647,0.403,0.9935); rgb(52pt)=(0.2621,0.4088,0.9946); rgb(53pt)=(0.2591,0.4145,0.9955); rgb(54pt)=(0.2556,0.4203,0.9965); rgb(55pt)=(0.2517,0.4261,0.9974); rgb(56pt)=(0.2473,0.4319,0.9983); rgb(57pt)=(0.2424,0.4378,0.9991); rgb(58pt)=(0.2369,0.4437,0.9996); rgb(59pt)=(0.2311,0.4497,0.9995); rgb(60pt)=(0.225,0.4559,0.9985); rgb(61pt)=(0.2189,0.462,0.9968); rgb(62pt)=(0.2128,0.4682,0.9948); rgb(63pt)=(0.2066,0.4743,0.9926); rgb(64pt)=(0.2006,0.4803,0.9906); rgb(65pt)=(0.195,0.4861,0.9887); rgb(66pt)=(0.1903,0.4919,0.9867); rgb(67pt)=(0.1869,0.4975,0.9844); rgb(68pt)=(0.1847,0.503,0.9819); rgb(69pt)=(0.1831,0.5084,0.9793); rgb(70pt)=(0.1818,0.5138,0.9766); rgb(71pt)=(0.1806,0.5191,0.9738); rgb(72pt)=(0.1795,0.5244,0.9709); rgb(73pt)=(0.1785,0.5296,0.9677); rgb(74pt)=(0.1778,0.5349,0.9641); rgb(75pt)=(0.1773,0.5401,0.9602); rgb(76pt)=(0.1768,0.5452,0.956); rgb(77pt)=(0.1764,0.5504,0.9516); rgb(78pt)=(0.1755,0.5554,0.9473); rgb(79pt)=(0.174,0.5605,0.9432); rgb(80pt)=(0.1716,0.5655,0.9393); rgb(81pt)=(0.1686,0.5705,0.9357); rgb(82pt)=(0.1649,0.5755,0.9323); rgb(83pt)=(0.161,0.5805,0.9289); rgb(84pt)=(0.1573,0.5854,0.9254); rgb(85pt)=(0.154,0.5902,0.9218); rgb(86pt)=(0.1513,0.595,0.9182); rgb(87pt)=(0.1492,0.5997,0.9147); rgb(88pt)=(0.1475,0.6043,0.9113); rgb(89pt)=(0.1461,0.6089,0.908); rgb(90pt)=(0.1446,0.6135,0.905); rgb(91pt)=(0.1429,0.618,0.9022); rgb(92pt)=(0.1408,0.6226,0.8998); rgb(93pt)=(0.1383,0.6272,0.8975); rgb(94pt)=(0.1354,0.6317,0.8953); rgb(95pt)=(0.1321,0.6363,0.8932); rgb(96pt)=(0.1288,0.6408,0.891); rgb(97pt)=(0.1253,0.6453,0.8887); rgb(98pt)=(0.1219,0.6497,0.8862); rgb(99pt)=(0.1185,0.6541,0.8834); rgb(100pt)=(0.1152,0.6584,0.8804); rgb(101pt)=(0.1119,0.6627,0.877); rgb(102pt)=(0.1085,0.6669,0.8734); rgb(103pt)=(0.1048,0.671,0.8695); rgb(104pt)=(0.1009,0.675,0.8653); rgb(105pt)=(0.0964,0.6789,0.8609); rgb(106pt)=(0.0914,0.6828,0.8562); rgb(107pt)=(0.0855,0.6865,0.8513); rgb(108pt)=(0.0789,0.6902,0.8462); rgb(109pt)=(0.0713,0.6938,0.8409); rgb(110pt)=(0.0628,0.6972,0.8355); rgb(111pt)=(0.0535,0.7006,0.8299); rgb(112pt)=(0.0433,0.7039,0.8242); rgb(113pt)=(0.0328,0.7071,0.8183); rgb(114pt)=(0.0234,0.7103,0.8124); rgb(115pt)=(0.0155,0.7133,0.8064); rgb(116pt)=(0.0091,0.7163,0.8003); rgb(117pt)=(0.0046,0.7192,0.7941); rgb(118pt)=(0.0019,0.722,0.7878); rgb(119pt)=(0.0009,0.7248,0.7815); rgb(120pt)=(0.0018,0.7275,0.7752); rgb(121pt)=(0.0046,0.7301,0.7688); rgb(122pt)=(0.0094,0.7327,0.7623); rgb(123pt)=(0.0162,0.7352,0.7558); rgb(124pt)=(0.0253,0.7376,0.7492); rgb(125pt)=(0.0369,0.74,0.7426); rgb(126pt)=(0.0504,0.7423,0.7359); rgb(127pt)=(0.0638,0.7446,0.7292); rgb(128pt)=(0.077,0.7468,0.7224); rgb(129pt)=(0.0899,0.7489,0.7156); rgb(130pt)=(0.1023,0.751,0.7088); rgb(131pt)=(0.1141,0.7531,0.7019); rgb(132pt)=(0.1252,0.7552,0.695); rgb(133pt)=(0.1354,0.7572,0.6881); rgb(134pt)=(0.1448,0.7593,0.6812); rgb(135pt)=(0.1532,0.7614,0.6741); rgb(136pt)=(0.1609,0.7635,0.6671); rgb(137pt)=(0.1678,0.7656,0.6599); rgb(138pt)=(0.1741,0.7678,0.6527); rgb(139pt)=(0.1799,0.7699,0.6454); rgb(140pt)=(0.1853,0.7721,0.6379); rgb(141pt)=(0.1905,0.7743,0.6303); rgb(142pt)=(0.1954,0.7765,0.6225); rgb(143pt)=(0.2003,0.7787,0.6146); rgb(144pt)=(0.2061,0.7808,0.6065); rgb(145pt)=(0.2118,0.7828,0.5983); rgb(146pt)=(0.2178,0.7849,0.5899); rgb(147pt)=(0.2244,0.7869,0.5813); rgb(148pt)=(0.2318,0.7887,0.5725); rgb(149pt)=(0.2401,0.7905,0.5636); rgb(150pt)=(0.2491,0.7922,0.5546); rgb(151pt)=(0.2589,0.7937,0.5454); rgb(152pt)=(0.2695,0.7951,0.536); rgb(153pt)=(0.2809,0.7964,0.5266); rgb(154pt)=(0.2929,0.7975,0.517); rgb(155pt)=(0.3052,0.7985,0.5074); rgb(156pt)=(0.3176,0.7994,0.4975); rgb(157pt)=(0.3301,0.8002,0.4876); rgb(158pt)=(0.3424,0.8009,0.4774); rgb(159pt)=(0.3548,0.8016,0.4669); rgb(160pt)=(0.3671,0.8021,0.4563); rgb(161pt)=(0.3795,0.8026,0.4454); rgb(162pt)=(0.3921,0.8029,0.4344); rgb(163pt)=(0.405,0.8031,0.4233); rgb(164pt)=(0.4184,0.803,0.4122); rgb(165pt)=(0.4322,0.8028,0.4013); rgb(166pt)=(0.4463,0.8024,0.3904); rgb(167pt)=(0.4608,0.8018,0.3797); rgb(168pt)=(0.4753,0.8011,0.3691); rgb(169pt)=(0.4899,0.8002,0.3586); rgb(170pt)=(0.5044,0.7993,0.348); rgb(171pt)=(0.5187,0.7982,0.3374); rgb(172pt)=(0.5329,0.797,0.3267); rgb(173pt)=(0.547,0.7957,0.3159); rgb(175pt)=(0.5748,0.7929,0.2941); rgb(176pt)=(0.5886,0.7913,0.2833); rgb(177pt)=(0.6024,0.7896,0.2726); rgb(178pt)=(0.6161,0.7878,0.2622); rgb(179pt)=(0.6297,0.7859,0.2521); rgb(180pt)=(0.6433,0.7839,0.2423); rgb(181pt)=(0.6567,0.7818,0.2329); rgb(182pt)=(0.6701,0.7796,0.2239); rgb(183pt)=(0.6833,0.7773,0.2155); rgb(184pt)=(0.6963,0.775,0.2075); rgb(185pt)=(0.7091,0.7727,0.1998); rgb(186pt)=(0.7218,0.7703,0.1924); rgb(187pt)=(0.7344,0.7679,0.1852); rgb(188pt)=(0.7468,0.7654,0.1782); rgb(189pt)=(0.759,0.7629,0.1717); rgb(190pt)=(0.771,0.7604,0.1658); rgb(191pt)=(0.7829,0.7579,0.1608); rgb(192pt)=(0.7945,0.7554,0.157); rgb(193pt)=(0.806,0.7529,0.1546); rgb(194pt)=(0.8172,0.7505,0.1535); rgb(195pt)=(0.8281,0.7481,0.1536); rgb(196pt)=(0.8389,0.7457,0.1546); rgb(197pt)=(0.8495,0.7435,0.1564); rgb(198pt)=(0.86,0.7413,0.1587); rgb(199pt)=(0.8703,0.7392,0.1615); rgb(200pt)=(0.8804,0.7372,0.165); rgb(201pt)=(0.8903,0.7353,0.1695); rgb(202pt)=(0.9,0.7336,0.1749); rgb(203pt)=(0.9093,0.7321,0.1815); rgb(204pt)=(0.9184,0.7308,0.189); rgb(205pt)=(0.9272,0.7298,0.1973); rgb(206pt)=(0.9357,0.729,0.2061); rgb(207pt)=(0.944,0.7285,0.2151); rgb(208pt)=(0.9523,0.7284,0.2237); rgb(209pt)=(0.9606,0.7285,0.2312); rgb(210pt)=(0.9689,0.7292,0.2373); rgb(211pt)=(0.977,0.7304,0.2418); rgb(212pt)=(0.9842,0.733,0.2446); rgb(213pt)=(0.99,0.7365,0.2429); rgb(214pt)=(0.9946,0.7407,0.2394); rgb(215pt)=(0.9966,0.7458,0.2351); rgb(216pt)=(0.9971,0.7513,0.2309); rgb(217pt)=(0.9972,0.7569,0.2267); rgb(218pt)=(0.9971,0.7626,0.2224); rgb(219pt)=(0.9969,0.7683,0.2181); rgb(220pt)=(0.9966,0.774,0.2138); rgb(221pt)=(0.9962,0.7798,0.2095); rgb(222pt)=(0.9957,0.7856,0.2053); rgb(223pt)=(0.9949,0.7915,0.2012); rgb(224pt)=(0.9938,0.7974,0.1974); rgb(225pt)=(0.9923,0.8034,0.1939); rgb(226pt)=(0.9906,0.8095,0.1906); rgb(227pt)=(0.9885,0.8156,0.1875); rgb(228pt)=(0.9861,0.8218,0.1846); rgb(229pt)=(0.9835,0.828,0.1817); rgb(230pt)=(0.9807,0.8342,0.1787); rgb(231pt)=(0.9778,0.8404,0.1757); rgb(232pt)=(0.9748,0.8467,0.1726); rgb(233pt)=(0.972,0.8529,0.1695); rgb(234pt)=(0.9694,0.8591,0.1665); rgb(235pt)=(0.9671,0.8654,0.1636); rgb(236pt)=(0.9651,0.8716,0.1608); rgb(237pt)=(0.9634,0.8778,0.1582); rgb(238pt)=(0.9619,0.884,0.1557); rgb(239pt)=(0.9608,0.8902,0.1532); rgb(240pt)=(0.9601,0.8963,0.1507); rgb(241pt)=(0.9596,0.9023,0.148); rgb(242pt)=(0.9595,0.9084,0.145); rgb(243pt)=(0.9597,0.9143,0.1418); rgb(244pt)=(0.9601,0.9203,0.1382); rgb(245pt)=(0.9608,0.9262,0.1344); rgb(246pt)=(0.9618,0.932,0.1304); rgb(247pt)=(0.9629,0.9379,0.1261); rgb(248pt)=(0.9642,0.9437,0.1216); rgb(249pt)=(0.9657,0.9494,0.1168); rgb(250pt)=(0.9674,0.9552,0.1116); rgb(251pt)=(0.9692,0.9609,0.1061); rgb(252pt)=(0.9711,0.9667,0.1001); rgb(253pt)=(0.973,0.9724,0.0938); rgb(254pt)=(0.9749,0.9782,0.0872); rgb(255pt)=(0.9769,0.9839,0.0805)}, mesh/rows=7]
table[row sep=crcr, point meta=\thisrow{c}] {%
x	y	c\\
1	30	nan\\
1	40	6.75692635751409\\
1	50	6.82802142645891\\
1	60	nan\\
1	70	6.93833506706009\\
1	80	6.95440284393294\\
1	90	7.02854478261959\\
1	100	7.08327706713118\\
1	110	7.12912124750448\\
1	120	7.17419957597233\\
1	130	7.18717513514141\\
1	140	7.31517426272508\\
1	150	7.37622512600363\\
1	160	7.33083818401956\\
1	170	7.48312272906612\\
1	180	7.50593181075458\\
1	190	7.55681290941868\\
1	200	7.61413294638494\\
1	210	7.68033147420177\\
1	220	7.71039140429452\\
1	230	7.73574304634172\\
1	240	7.77845743062446\\
1	250	7.81461368832096\\
1	260	7.85614450437473\\
1	270	7.91619476261491\\
1	280	7.95452715967809\\
1	290	7.93984975292145\\
1	300	7.93762567522509\\
2	30	6.5904673443022\\
2	40	nan\\
2	50	nan\\
2	60	6.88033117484364\\
2	70	6.92475620098188\\
2	80	6.95738775477086\\
2	90	nan\\
2	100	7.18470347552481\\
2	110	7.0945102581989\\
2	120	7.16639163485858\\
2	130	7.17480164620977\\
2	140	7.29567235020999\\
2	150	7.36753247104057\\
2	160	7.4325358722183\\
2	170	7.47136398105109\\
2	180	7.50386193689214\\
2	190	7.55160297158169\\
2	200	7.46869639311739\\
2	210	7.67782087992319\\
2	220	7.70672721331643\\
2	230	7.71668365704351\\
2	240	7.77198118820776\\
2	250	7.80434342559749\\
2	260	7.85150289833694\\
2	270	7.88284516204192\\
2	280	7.94713747506694\\
2	290	7.93648440434283\\
2	300	7.97517022768374\\
3	30	nan\\
3	40	6.69729411348347\\
3	50	6.78854184425828\\
3	60	6.89854197451777\\
3	70	6.92581375353562\\
3	80	6.95392912921633\\
3	90	7.04448832958661\\
3	100	7.08030181451088\\
3	110	7.08377497965577\\
3	120	7.15071432826859\\
3	130	7.24758537429995\\
3	140	7.29479987480323\\
3	150	7.34974121619789\\
3	160	7.32966846576607\\
3	170	7.47121306902735\\
3	180	7.51275899167529\\
3	190	7.54425137894964\\
3	200	7.61018161689837\\
3	210	7.60782292193054\\
3	220	7.71328190336285\\
3	230	7.68789858319324\\
3	240	7.75291253948199\\
3	250	7.79399469926671\\
3	260	7.84765206014462\\
3	270	7.90967480056242\\
3	280	7.94212610371846\\
3	290	7.96029945676611\\
3	300	7.96281802184881\\
4	30	nan\\
4	40	nan\\
4	50	6.65874460478732\\
4	60	6.78376644202847\\
4	70	6.87807597192913\\
4	80	6.95216541709851\\
4	90	7.01052217557747\\
4	100	7.07263439078991\\
4	110	7.12336132035192\\
4	120	7.18095128559347\\
4	130	7.24591704908786\\
4	140	7.29974423744974\\
4	150	7.34248476478471\\
4	160	7.42222459853284\\
4	170	7.42972196434094\\
4	180	7.49896800460901\\
4	190	7.54889969373032\\
4	200	7.6079041927661\\
4	210	7.68040650140628\\
4	220	7.69997662045566\\
4	230	7.70560387886403\\
4	240	7.77090845348129\\
4	250	7.78224737886802\\
4	260	7.84831695937837\\
4	270	7.90866520035625\\
4	280	7.95381796318539\\
4	290	7.92784551793412\\
4	300	7.97011762042159\\
5	30	nan\\
5	40	nan\\
5	50	nan\\
5	60	nan\\
5	70	nan\\
5	80	6.82458968004241\\
5	90	6.91257858379575\\
5	100	7.00391984231441\\
5	110	7.05316746066102\\
5	120	7.1390581831084\\
5	130	7.20383000366348\\
5	140	7.26637416350056\\
5	150	7.33379436204632\\
5	160	7.41825840224335\\
5	170	7.47190288239252\\
5	180	7.48557069693048\\
5	190	7.56314088982565\\
5	200	7.62422281676148\\
5	210	7.67093266176893\\
5	220	7.70652909266891\\
5	230	7.75190629086967\\
5	240	7.75874548141129\\
5	250	7.8113166005042\\
5	260	7.86675822650289\\
5	270	7.91252618114772\\
5	280	7.93764255276302\\
5	290	7.97194680354858\\
5	300	8.0015783962664\\
6	30	nan\\
6	40	nan\\
6	50	nan\\
6	60	nan\\
6	70	nan\\
6	80	nan\\
6	90	nan\\
6	100	nan\\
6	110	6.61808721321667\\
6	120	7.02136006490629\\
6	130	7.10071965177456\\
6	140	7.11663552832721\\
6	150	7.25048954936599\\
6	160	7.30435503586581\\
6	170	9.00005578423109\\
6	180	11.3117622206068\\
6	190	9.27433287746074\\
6	200	9.81825482442505\\
6	210	8.28737491527797\\
6	220	7.64142801815771\\
6	230	7.63261705899451\\
6	240	7.70880644846013\\
6	250	7.76081581361659\\
6	260	7.79764451602426\\
6	270	7.82073022080485\\
6	280	7.87020932176696\\
6	290	7.89691171290557\\
6	300	7.78002773743892\\
7	30	nan\\
7	40	nan\\
7	50	nan\\
7	60	nan\\
7	70	nan\\
7	80	nan\\
7	90	nan\\
7	100	nan\\
7	110	nan\\
7	120	nan\\
7	130	nan\\
7	140	nan\\
7	150	nan\\
7	160	nan\\
7	170	8.09973116671654\\
7	180	7.63184557618456\\
7	190	7.48774842415723\\
7	200	7.20821651344286\\
7	210	7.26154151040364\\
7	220	7.21576309907279\\
7	230	nan\\
7	240	7.61316219849436\\
7	250	7.76882697981076\\
7	260	7.81884381039101\\
7	270	7.85502851646835\\
7	280	7.90686360738113\\
7	290	7.9701960141276\\
7	300	8.00005240149958\\
};
\end{axis}

\begin{axis}[%
width=1in,
height=0.951in,
at={(0in,0in)},
scale only axis,
point meta min=0,
point meta max=1,
xmin=0,
xmax=1,
ymin=0,
ymax=1,
axis line style={draw=none},
ticks=none
]
\end{axis}
\end{tikzpicture}%